\documentclass[reprint, aps, pra, amsfonts, amssymb, amsmath,  showkeys, superscriptaddress, twocolumn, longbibliography]{revtex4-2}

\usepackage[T1]{fontenc}
\usepackage[british]{babel}

\usepackage{silence}
\usepackage{float}
\usepackage{booktabs}
\usepackage{extarrows}
\usepackage{tikz-cd}
\usepackage{amsmath,amssymb,amsfonts,amscd,amsthm}
\usepackage[shortlabels]{enumitem}
\usepackage{algorithm}
\usepackage{algpseudocode}

\usepackage[colorlinks=true,citecolor=purple,linkcolor=purple,urlcolor=purple, bookmarksdepth=section]{hyperref}

\usepackage{url}

\usepackage{upgreek}

\newtheorem{theorem}{Theorem}
\newtheorem{lemma}{Lemma}

\theoremstyle{plain}
\newtheorem{definition}{Definition}

\usepackage{orcidlink}

\newcommand{\expect}[1]{\left\langle#1\right\rangle}
\newcommand{\ket}[1]{\vert {#1}\rangle}
\newcommand{\bra}[1]{\langle {#1}\vert}
\newcommand{\tr}[1]{\text{tr}\left\{#1\right\}}

\newcommand{\um}{\mathfrak{u}(m)}
\newcommand{\imum}{d\varphi(\mathfrak{u}(m))}

\newcommand{\uD}{\mathfrak{u}(D)}
\newcommand{\Um}{\mathrm{U}(m)}
\newcommand{\UD}{\mathrm{U}(D)}

\newcommand{\dphi}{d\varphi}
\newcommand{\dphiinv}{d\varphi^{-1}}
\newcommand{\hnm}{\mathcal{F}^m_n}
\newcommand{\nablag}{\nabla_{\! S\,} g}

\newcommand{\inket}{\ket{\rm in}}
\newcommand{\inbra}{\bra{\rm in}}
\newcommand{\tket}{\ket{t}}
\newcommand{\tbra}{\bra{t}}
\newcommand{\hket}{\ket{h}}

\newcommand{\thket}{\ket{t,h}}
\newcommand{\thbra}{\bra{t,h}}
\newcommand{\outket}{\ket{\rm out}}
\newcommand{\outbra}{\bra{\rm out}}
\newcommand{\Der}{\mathrm{D}}

\allowdisplaybreaks

\begin{document}

\title{Riemannian optimization for linear optical problems}

\author{Pablo V. Parellada \orcidlink{0000-0002-6768-6671}}
\thanks{pablo.veganzones@uva.es}
\affiliation{Departamento de Física Teórica, Atómica y Óptica, Universidad de Valladolid, Spain}
\affiliation{Laboratory for Disruptive Interdisciplinary Science (LaDIS), Universidad de Valladolid, Spain.}
\date{17th September 2026}

\begin{abstract}

Many applications of linear optics in quantum science involve searching for some optimal configuration of a multiport interferometer. For example, to find heralded state preparations, one can optimize the interferometer to maximize the fidelity and success probability of the state preparation. In this paper, we derive a closed analytical formula for the Riemannian gradient of any differentiable function defined over the interferometer unitaries. We use this gradient in Riemannian optimization algorithms, such as gradient descent or BFGS, to find the optimal setup of the interferometer. As an interesting use case, we search for linear optical state preparations, improving some of the success probabilities reported in the literature for NOON or photon catalysis preparations. Finally, we show that our optimizer is orders of magnitude faster than other methods in the literature, opening the door for solving linear optical problems involving more modes and photons.

\end{abstract}

\maketitle

\section{Introduction}

The simplicity of passive multiport interferometers makes them an essential tool in quantum technologies \cite{tan_resurgence_2019}, with applications in quantum computing \cite{kok_linear_2007}, quantum communications \cite{duan_long-distance_2001}, quantum metrology \cite{barbieri_optical_2022} and quantum foundations \cite{genovese_research_2005}. However, this simplicity limits their applications: linear optics can only implement a small subset of all the possible unitary transformations in Fock space \cite{moyano-fernandez_linear_2017, oszmaniec_universal_2017}. For example, unitary gates essential for quantum computing, like the CNOT, cannot be implemented deterministically \cite{knill_scheme_2001, knill_quantum_2002, carolan_universal_2015}. Moreover, many useful entangled states, like Bell or NOON states, cannot be generated deterministically from a photon number state \cite{parellada_no-go_2023, gliniasty_simple_2024, parellada_lie_2026, singh_rigorous_2026}. This critical limitation can be circumvented by allowing probabilistic preparations conditioned on measuring a herald in some auxiliary modes. However, finding heralded preparations analytically is a hard problem, solved only for specific families of states \cite{aniello_engineering_2007, Chin_exponentially_2024, Bhatti_2025, Kang_2026}, and in general one resorts to optimization techniques \cite{forbes_heralded_2025}.


There are two main kinds of optimization algorithms found in quantum optics: discrete and continuous. Discrete algorithms, such as genetic or discrete learning algorithms, search for optimal combinations of gates from a certain toolbox. Some applications of discrete algorithms are related to the discovery of new quantum experiments \cite{krenn_automated_2016}, quantum metrology \cite{knott_search_2016}, generation of graph states \cite{Lee2023graphtheoretical, lin_graphiq_2024}, gate preparation \cite{Chernikov_heralded,zhang_variational_2024} or the study of entanglement \cite{lachman_certification_2026}. Continuous algorithms encode the problem in a cost function and make incremental updates to optimize it. Most of them rely on computing or approximating the gradient of the cost function, which gives the direction of steepest descent (or some variant of this method). The gradient is usually computed analytically or numerically on a classical computer, although it can sometimes be computed directly in the photonic circuit using parameter-shift rules \cite{exact_gradients_linear_optics2024, vqc_photonic_param_shift20225}. Continuous optimization has been applied to the study of entanglement \cite{stanisic_generating_2017}, state preparation \cite{gubarev_improved_2020, gubarev_fock_2021, fldzhyan_compact_2021, yao_riemannian_2024, hartnett_automated_2026, aralov_photon_2026}, gate preparation \cite{uskov_maximal_2009, arrazola_machine_2019, garcia-escartin_optimal_2021}, the discovery of new quantum experiments \cite{pytheus} or the training of quantum machine learning models \cite{chabaud_quantum_2021}.

In the literature, the most common approach for optimizing a passive photonic circuit is to decompose it into parametrized beam splitters and phase shifters (usually using the Clements decomposition \cite{clements_optimal_2016}). These parameters are then optimized to minimize some cost function, like the fidelity or success probability of a heralded state preparation \cite{stanisic_generating_2017, gubarev_improved_2020, fldzhyan_compact_2021, gubarev_fock_2021}. However, parametrized matrices can lead to cost functions with complicated landscapes. In recent decades, Riemannian optimization has gained traction as an alternative for optimizing functions on matrix manifolds \cite{absil_optimization_2008, boumal_introduction_2023} by directly updating the matrices without any parametrization. Given the prevalence of unitary matrices in quantum physics, Riemannian optimization has made its way into the field \cite{schulte_gradient_2010, luchnikov_qgopt_2021, luchnikov_riemannian_2021, quantum_classical_eigensolver}, including recent applications to training quantum circuits \cite{quantum_riemannian_flow} and Gaussian photonic circuits \cite{yao_riemannian_2024}.

In this work, we derive an analytical formula for the Riemannian gradient of any differentiable function defined over passive interferometers (unitary matrices). Additionally, we implement fast and efficient gradients for the fidelity and success probability of heralded state preparations. These gradients are then used by Riemannian optimization algorithms, like gradient descent or BFGS, which we implemented in QOptCraft \cite{aguado_qoptcraft_2023, qoptcraft_repo}. Our implementation has several advantages over other optimization methods found in the literature: (1) it optimizes directly on the unitary manifold, avoiding parametrizations of the unitary group that worsen the cost function's landscape; (2) it uses a closed analytical formula for the gradient, avoiding automatic differentiation; (3) it is orders of magnitude faster than existing optimizers; (4) in addition to first-order methods, it implements Riemannian quasi-Newton methods for faster convergence. These advantages allow us to improve on some of the success probabilities of state preparations found in the literature.

The paper is organized as follows. In Section \ref{sec: riemannian optimization}, we review the fundamentals of Riemannian optimization. In Section \ref{sec: linear optics}, we review the Lie group description of passive linear optics and use it to derive the closed-form analytical gradient of a linear optical function. In Section \ref{sec: examples}, we apply Riemannian optimization to interesting problems in linear optics, such as finding heralded state and gate preparations. In Section \ref{sec: benchmarks}, we provide some details about the implementation of our Riemannian optimizer and we benchmark it on several heralded state preparations against other optimizers proposed in the literature. In Section \ref{sec: discussion}, we discuss the results and draw the final conclusions. 

Additionally, in Appendix \ref{appendix: algorithms} we review more advanced Riemannian optimization algorithms, like quasi-Newton or momentum methods, and in Appendix \ref{appendix: invariants} we propose a faster algorithm for computing exact state preparations using Lie algebraic invariants.

\section{Riemannian optimization}\label{sec: riemannian optimization}

The objective of Riemannian optimization is to minimize or maximize functions on Riemannian manifolds. It avoids any parametrization of the manifold and, instead, updates the points directly on the manifold, tracing (approximately) geodesic paths. Its effectiveness is motivated by the observation that some functions are geodesically convex even if they are not convex at all, improving the optimization landscape. The theory of Riemannian optimization draws heavily from differential geometry \cite{lee_introduction_2013}, which exceeds the scope of this article, so we refer the reader to the standard references \cite{absil_optimization_2008, boumal_introduction_2023} for rigorous introductions.

In this section, we first provide a quick summary of the main tools from differential geometry that we shall need. Then, we explain how the Riemannian gradient descent algorithm works and we comment on other algorithms, which we detail in Appendix \ref{appendix: algorithms}.

\subsection{Manifolds and Lie groups}\label{sec: manifolds}

Intuitively, a smooth manifold is a set of points that can be parametrized locally by smooth ($\mathcal{C}^\infty$) functions. Unitary matrices form a manifold because they can be parametrized smoothly (for example, via the Reck \cite{reck_experimental_1994} or Clements \cite{clements_optimal_2016} decompositions). Unitary matrices also have a group structure: products and inverses of unitary matrices are unitary, the identity is unitary and the associative property holds. In fact, unitary matrices form a Lie group: a smooth manifold with a group structure such that the product and inverse of unitary matrices are smooth maps \cite{lee_introduction_2013}.

At each point, $x$, of a smooth manifold, $\mathcal{M}$, we can define the tangent space, $T_x \mathcal{M}$, as the vector space of velocities (at point $x$) of smooth curves in $\mathcal{M}$ passing through $x$. The union of all tangent spaces is also a smooth manifold called the tangent bundle, $T\mathcal{M}$, and its elements are $(x, v)\in T\mathcal{M}$, with $x\in \mathcal{M}$ and $v\in T_x \mathcal{M}$.

A particularly important example of a tangent space is the tangent space of a Lie group at the identity, $T_e \mathcal{G}$, which is nothing but its \textit{Lie algebra}. The group structure of the manifold provides an isomorphism between its tangent bundle $T\mathcal{G}$ and $T_e \mathcal{G} \times \mathcal{G}$ \cite{lee_introduction_2013}, which maps an \textit{a priori} complicated manifold into a simple Cartesian product. In matrix Lie groups, this isomorphism can be realized with the so-called \textit{left-trivialization}, which defines the tangent space at an arbitrary point $g \in \mathcal{G}$ as the products of $g$ with vectors in the Lie algebra: $T_g \mathcal{G} = \{g X \:| \: X \in T_e \mathcal{G}\}$.

Since tangent spaces are vector spaces, they can be endowed with an inner product, $\langle \cdot, \cdot\rangle_x$. If these inner products are defined on the whole tangent bundle and vary smoothly with $x$, we obtain a tensor field of inner products called the \textit{Riemannian metric} \cite[Def. 3.52]{boumal_introduction_2023}. A smooth manifold with such a metric is called a \textit{Riemannian manifold}.

Let us illustrate these concepts with an essential example in quantum theory: the Lie group of $m \times m$ unitary matrices, $\Um$. Its Lie algebra is the vector space of anti-Hermitian matrices, $\um = \{X \in \mathbb{C}^{m\times m} \:|\: X = -X^\dagger\}$. The (left-trivialized) tangent space at an arbitrary unitary $U$ is $T_U \Um=\{U X \:|\: X \in \um\}$. In each tangent space, we can define the Hilbert-Schmidt inner product $\langle X, Y\rangle=\tr{X^\dagger Y}$, which defines a metric on the unitary group. In fact, the left-trivialization reduces the inner product of two vectors, $UX, UY \in T_U \Um$, to the inner product of their corresponding vectors in the Lie algebra, $X, Y \in \um$: $\langle UX, UY\rangle = \tr{(UX)^\dagger (UY)}=\tr{X^\dagger Y}$.

\subsection{Functions on manifolds}

We can define functions that map points in a manifold to another manifold, $f: \mathcal{M} \rightarrow \mathcal{M}'$. The differential of a smooth map $f:\mathcal{M} \rightarrow \mathcal{M}'$ at a point $x\in \mathcal{M}$ is the linear map $\Der f(x):T_x\mathcal{M}\rightarrow T_{f(x)}\mathcal{M}'$ defined by
\begin{equation}\label{eq: differential}
    \Der f(x)[v] = \frac{d}{dt}f(\gamma(t)) \bigg|_{t=0}
\end{equation}
for any smooth curve $\gamma$ such that $\gamma(0)=x$ and $\gamma'(0)=v$. An important property of the differential is the \textit{chain rule}. Suppose we have another function between manifolds, $g: \mathcal{N} \rightarrow \mathcal{M}$. Then, the derivative of $f \circ g: \mathcal{N} \rightarrow \mathcal{M}'$ at $y\in \mathcal{N}$ in direction $w$ is
\begin{equation}\label{eq: chain rule}
    \Der (f\circ g)(y)[w] = \Der f(g(y))[\Der g(y)[w]] \:.
\end{equation}

Of particular importance are real functions, $f: \mathcal{M} \rightarrow \mathbb{R}$. For example, if $A$ and $B$ are Hermitian matrices, $f(U) = \tr{UAU^\dagger B}$ is a smooth real function on the unitary group. The \textit{Riemannian gradient} of a real function $f$ is then defined as the vector field $\nabla f$ such that for every $(x,v)\in T\mathcal{M}$:
\begin{equation}\label{eq: riemannian gradient}
    \Der f(x)[v] = \langle v, \nabla f(x)\rangle_x \:,
\end{equation}
where $\langle \cdot, \cdot\rangle_x$ is the Riemannian metric evaluated at $x$ \cite[Def. 3.58]{boumal_introduction_2023}. That means that every $\nabla f(x)$ is a tangent vector belonging to $T_x \mathcal{M}$.

\subsection{Riemannian gradient descent}\label{sec: rgd}

One simple approach to minimizing a real function is to select an arbitrary point of the domain and move it along the path traced by the (negative) gradient, which is the direction of the steepest descent of the function. When the gradient is zero, a local minimum is found. If the domain of the function is $\mathbb{R}^n$, the numerical algorithm for (Euclidean) gradient descent is given by
\begin{equation}\label{eq: eucl grad descent}
    x_{n+1} = x_n - \alpha \nabla f(x_n) \:,
\end{equation}
where $x_n \in \mathbb{R}^n$ is a point in the domain and $\alpha$ is a real number that tunes the step size of the parameter update. 

Since unitary matrices are not a vector space, we cannot naively apply the update \eqref{eq: eucl grad descent} directly to some unitary $U_n$. One approach to circumvent this problem is to parametrize unitary matrices and then apply gradient descent to those parameters. However, we can avoid parametrizing the unitaries altogether by using a gradient descent algorithm that maps the current unitary matrix to another unitary with a lower value of the function. This is the objective of \textit{Riemannian gradient descent}.

Recall that the Riemannian gradient at a point $x$, $\nabla f(x)$, is a tangent vector, so we need a way to map vectors in the tangent space back to points on the manifold. These special maps are called retractions \cite{absil_optimization_2008, boumal_introduction_2023}:
\begin{definition}
    A retraction on a manifold $\mathcal{M}$ is a smooth map $R$ from the tangent bundle $T\mathcal{M}$ onto $\mathcal{M}$
    \begin{equation*}
    \begin{aligned}
    R: \:\: T\mathcal{M} & \rightarrow \mathcal{M} \\
    (x,v) &\mapsto R_x(v)
    \end{aligned}    
    \end{equation*}
    such that any curve $\gamma(t) = R_x(tv)$ satisfies $\gamma(0)=x$ and $\dot\gamma(0)=v$.
\end{definition}

Several examples of retractions can be given for the unitary group. In our implementation of Riemannian optimization, we use by default the geometric exponential retraction, which is given by 
\begin{equation}\label{eq: exp retraction}
    R_U(UX) = U\exp(X) \:,
\end{equation}
where $(U, UX) \in T \Um$ and $\exp(\,\cdot\,)$ is the matrix exponential. In problems where the unitary matrices are very large (not our case), cheaper retractions like the Cayley transform, $R_U(UX) = U(\mathrm{Id} - \frac{1}{2}X)^{-1}(\mathrm{Id}+\frac{1}{2}X)$, can be used. Other examples of retractions can be found in \cite[Example 4.1.2]{absil_optimization_2008} for the orthogonal group (the proofs with unitaries are identical).


The fundamental algorithm for minimizing functions on a Riemannian manifold is \textit{Riemannian gradient descent} (Algorithm \ref{alg: RGD}). It is analogous to Euclidean gradient descent \eqref{eq: eucl grad descent}, but uses a retraction to map the current point and gradient into another point in the manifold. Its convergence to a local minimum of the cost function is at least linear \cite{boumal_introduction_2023}. If the retraction is the geometric exponential, the path described by Riemannian gradient descent is a geodesic in the manifold. This gives an advantage over Euclidean optimization because some functions are geodesically convex but not convex.

\begin{algorithm}[H]
\caption{Riemannian gradient descent \cite{boumal_introduction_2023}}\label{alg: RGD}
\begin{algorithmic}[1]
\State \textbf{Input:} $x_0 \in \mathcal{M}$
\For{$k = 0, 1, 2, \ldots$}
    \State Pick a step size $\alpha_k > 0$
    \State $s_k = -\alpha_k \operatorname{grad} f(x_k)$
    \State $x_{k+1} = \mathrm{R}_{x_k}(s_k)$
\EndFor
\State \textbf{Output:} Local minimum of $f$
\end{algorithmic}
\end{algorithm}

The step size $\alpha_k$ can be chosen in many ways. A common strategy is to use the so-called \textit{backtracking line-search} (Algorithm \ref{alg: line search}), which picks a step size that satisfies the Armijo-Goldstein condition \cite{boumal_introduction_2023}.

\begin{algorithm}[H]
\caption{Backtracking line-search \cite{boumal_introduction_2023}}\label{alg: line search}
\begin{algorithmic}[1]
\Require $\tau, r \in (0,1)$; for example, $\tau = \frac{1}{2}$ and $r = 10^{-4}$.
\State \textbf{Input:} $x \in \mathcal{M}$, $\bar{\alpha} > 0$.
\State Set $\alpha \leftarrow \bar{\alpha}$
\While{$f(x) - f(\mathrm{R}_x(-\alpha \operatorname{grad} f(x))) < r\alpha \|\operatorname{grad} f(x)\|^2$}
    \State Set $\alpha \leftarrow \tau\alpha$
\EndWhile  \Comment{Armijo-Goldstein condition is satisfied}
\State \textbf{Output:} $\alpha$
\end{algorithmic}
\end{algorithm}

Finally, we refer the reader to Appendix \ref{appendix: algorithms} for more advanced Riemannian optimization algorithms, such as momentum algorithms (e.g., Nesterov or ADAM) and quasi-Newton methods (e.g., BFGS or L-BFGS).

\section{Gradients of linear optical functions}\label{sec: linear optics}

In this section, we introduce the basic description of passive linear optics using the language of differential geometry, such as Lie groups, homomorphisms and Riemannian manifolds. After reviewing the basic concepts, we introduce the main result of the paper: an analytical formula for the gradient of a differentiable function defined on the manifold of passive linear interferometers. This formula can be plugged into the Riemannian optimization algorithms from Section \ref{sec: rgd} to solve interesting problems in linear optics, as shown in Section \ref{sec: examples}.

\subsection{Lie groups and linear optics}\label{sec: lie groups optics}

Passive linear interferometers, like the Mach-Zehnder interferometer, are optical circuits in which modes of light can propagate and interact via beam splitters and phase shifters. Mathematically, an interferometer with $m$ modes is modeled by an $m\times m$ unitary matrix, $S$, sometimes called the \textit{scattering matrix}. In classical optics, the scattering matrix describes the evolution of the amplitude of each mode of the electromagnetic field through the interferometer. In quantum optics, it describes the evolution of the creation and annihilation operators of each mode of the field \cite{kok_linear_2007,skaar_quantum_2004}.

More precisely, the evolution of a multiphoton state in Fock space, $\ket{\psi} \in \mathcal{F}$, through a linear interferometer with scattering matrix $S$ is given by $\varphi(S)\ket{\psi}$, where $\varphi: \Um \rightarrow \mathrm{U}(\mathcal{F})$ is a unitary representation of the group of scattering matrices on Fock space. Being a representation, $\varphi$ is in particular a group homomorphism,
\begin{equation}\label{eq: homomorphism property}
    \varphi(S_2 S_1) = \varphi(S_2) \varphi(S_1) \:, 
\end{equation}
which physically means that evolving a quantum state through an interferometer $S = S_2 S_1$ is the same as evolving it first through $S_1$ and then through $S_2$.

Since passive interferometers preserve the total number of photons, $\varphi$ leaves invariant each subspace of Fock space with $n$ total photons, $\mathcal{F}^m_n = \mathrm{span}\left(\{\ket{n_1 \ldots n_m} : n_i \geq 0, \sum_i n_i = n\}\right)$, which has dimension $D = \binom{m+n-1}{n}$. These subspaces
are precisely the irreducible subrepresentations into which $\varphi$ decomposes \cite{aniello_exploring_2006}, so we can work in one photon-number subspace at a time. Restricted to $\mathcal{F}^m_n$, $\varphi: \Um \rightarrow \UD$ is a homomorphism between finite-dimensional unitary groups. It is also a smooth map, since the entries of $\varphi(S)$ are permanents of submatrices of $S$ \cite{scheel_permanents_2004, aaronson_computational_2010}.

Since $\varphi$ is smooth, we can compute its differential (which will appear later on when computing the gradient of cost functions). The differential of a Lie group homomorphism at the identity, $\dphi \equiv D\varphi(\mathrm{Id})$, is a Lie algebra homomorphism \cite{garcia-escartin_multiple_2019}:
\begin{equation}\label{eq: algebra homomorphism}
    \begin{aligned}
        d\varphi: \:\: \um \:&\rightarrow\: \imum \subset \uD \\
        G \: &\mapsto\: \sum_{j k} G_{j k} a_j^\dagger a_k  \:,
    \end{aligned}
\end{equation}
where $\um$ is the Lie algebra of $\Um$, $\imum$ is the linear optical subalgebra of $\uD$ and the operator $\dphi(G)=\sum_{j k} G_{j k} a_j^\dagger a_k$ is restricted to $\mathcal{F}^m_n$.

Finally, recall that these algebras are inner product spaces with the Hilbert-Schmidt inner product $\langle X, Y\rangle=\tr{X^\dagger Y}$ (Section \ref{sec: manifolds}).

\subsection{Computing linear optical gradients}

In this section we provide the main result of this paper: an analytical formula for the gradient of a real function defined on the manifold of linear optical unitaries. This gradient can be used by any Riemannian optimization algorithm, like gradient descent or BFGS, to solve problems in linear optics, as we show in Section \ref{sec: examples}.

First, suppose we have a real function $f: \UD \rightarrow \mathbb{R}$, where $\UD$ is the manifold of all quantum evolution operators (not only passive linear optical ones) acting on $\hnm$. This function will play the role of an objective function that we want to minimize. Its Riemannian gradient with respect to $U\in \UD$ is a vector field $\nabla_U f$ given by Eq. \eqref{eq: riemannian gradient}:
\begin{equation}\label{eq: der cost U}
    \Der f(U)[UX] = \expect{UX, \nabla_U f(U)} \:,
\end{equation}
for all vector directions $UX \in T_U \UD$ (or, equivalently, $X \in \uD$).

Now suppose that $U$ is not an arbitrary unitary evolution, but a passive linear optical one; that is, suppose that $U=\varphi(S)$, where $S$ is the matrix of the interferometer. Composing $f(U)$ with $U=\varphi(S)$ results in a new function $g \equiv f \circ \varphi: \Um \rightarrow \mathbb{R}$ defined over the manifold of scattering matrices. To compute the gradient of $g$ with respect to $S\in \Um$, we first compute its differential using the chain rule \eqref{eq: chain rule}:
\begin{equation}
    \Der g(S)[SY] = \Der f(\varphi(S))[\Der \varphi(S)[SY]]
\end{equation}
where $SY \in T_S \Um$ is a vector direction (or, equivalently, $Y\in \um$). Substituting $Df$ from Eq. \eqref{eq: der cost U}:
\begin{equation}\label{eq: der cost S}
    \Der g(S)[SY] = \expect{\Der \varphi(S)[SY], \nabla_U f(\varphi(S)) } \:.
\end{equation}

We simplify $\Der \varphi(S)[SY]$ using the definition of differential \eqref{eq: differential}, for a smooth curve $\gamma(t) = S \exp(t Y)$, and the homomorphism property of $\varphi$ \eqref{eq: homomorphism property}:
\begin{equation}\label{eq: differential group homomorphism}
    \begin{aligned}
    \Der \varphi (S)[SY] &= \frac{d}{dt} \bigg|_{t=0} \varphi(S \exp(t Y)) \\[1mm]
    &= \varphi(S) \,\frac{d}{dt} \bigg|_{t=0} \varphi(\exp(t Y))  \\[1mm]
    &= \varphi(S) \Der \varphi(\mathrm{Id})[Y]   \\[1mm]
    &= \varphi(S) d\varphi(Y) \:,
    \end{aligned}
\end{equation}
where the last step is simply the definition of the algebra homomorphism, $d\varphi$, as the differential of $\varphi$ at the identity.

Substituting Eq. \eqref{eq: differential group homomorphism} in Eq. \eqref{eq: der cost S}, we obtain
\begin{equation}\label{eq: der cost S 1}
    \begin{aligned}
    \Der g(S)[SY] &= \expect{\varphi(S) d\varphi(Y), \nabla_U f(\varphi(S)) } \\[1mm]
    &= \expect{d\varphi(Y), \varphi(S)^\dagger \nabla_U f(\varphi(S)) } \:.
    \end{aligned}
\end{equation}
We can isolate $Y$ by introducing the adjoint map of $\dphi$, i.e. a linear map $\dphi^\ast:\uD \rightarrow \um $ such that $\expect{\dphi(Y), X} = \expect{Y, \dphi^\ast(X)}$ for every $Y\in \um$ and every $X\in \uD$. The adjoint map can be computed easily after choosing an orthonormal basis, as the following lemma shows:
\begin{lemma}\label{lemma: adjoint map}
    Let $\{G_i\}$ be an orthonormal basis of $\um$ and let $\dphi: \um \rightarrow \uD$ be the algebra homomorphism given by Eq. \eqref{eq: algebra homomorphism}. The adjoint map of $\dphi$, denoted by $\dphi^\ast:\uD \rightarrow \um$, is given by:
    \begin{equation}\label{eq: adjoint algebra homomorphism}
        \dphi^\ast(X) = \sum_{i=1}^{m^2} \expect{\dphi(G_i), X} G_i
    \end{equation}
    for every $X\in \uD$.
\end{lemma}
\begin{proof}
    The algebra homomorphism can be computed from the orthonormal basis $\{G_i\}$ as
    \begin{equation}\label{eq: dphi coords}
        \dphi(Y) = \dphi(\textstyle\sum_i \expect{Y, G_i} G_i) = \sum_i \expect{Y, G_i} \dphi(G_i) \:.
    \end{equation}
    Given an inner product $\expect{\dphi(Y), X}$, we can rearrange the terms to isolate the formula for the adjoint map $\dphi^\ast$:
    \begin{equation}
        \begin{aligned}
            \expect{\dphi(Y), X} &=\expect{\textstyle\sum_i \expect{Y, G_i} \dphi(G_i), X} \\[1mm]
            &=\textstyle\sum_i \expect{Y, G_i} \expect{\dphi(G_i), X} \\[1mm]
             &= \expect{Y, \textstyle\sum_i\expect{\dphi(G_i), X} G_i} \\[1mm]
             &:= \expect{Y, \dphi^\ast(X)} \:.
        \end{aligned}
    \end{equation}
    Thus, the adjoint is given by Eq. \eqref{eq: adjoint algebra homomorphism}. 
\end{proof}
Finally, applying Lemma \ref{lemma: adjoint map} to Eq. \eqref{eq: der cost S 1}, the vector direction $SY$ is isolated from the derivative of $g$:
\begin{equation}\label{eq: der cost S 2}
    \begin{aligned}
    \Der g(S)[SY] &= \expect{d\varphi(Y), \varphi(S)^\dagger \nabla_U f(\varphi(S))} 
    \\[1mm]
     &=\expect{Y, \dphi^\ast\left(\varphi(S)^\dagger \nabla_U f(\varphi(S)) \right)}
    \\[1mm]
     &=\expect{SY, S\dphi^\ast\left(\varphi(S)^\dagger \nabla_U f(\varphi(S)) \right)}.
    \end{aligned}
\end{equation}
Thence, we obtain the Riemannian gradient of $g$ on the right-hand side of the inner product (recall Eq. \eqref{eq: riemannian gradient}). We wrap up this important result in the following theorem:
\begin{theorem}\label{thm: grad S}
Let $f: \UD \rightarrow \mathbb{R}$ be a real differentiable function with gradient $\nabla_U f$. Let $\varphi: \Um \rightarrow \UD$ be the $n$-photon homomorphism and $d\varphi$ its associated algebra homomorphism. Then, the Riemannian gradient of $g=f\circ \varphi$ is
\begin{equation}\label{eq: grad cost S}
    \nablag (S) = S\dphi^\ast\Big(\varphi(S)^\dagger \nabla_U f(\varphi(S)) \Big) \:,
\end{equation}
where $\dphi^\ast$, given by Eq. \eqref{eq: adjoint algebra homomorphism}, is the adjoint of $\dphi$.
\end{theorem}

\subsection{Gradient examples}\label{sec: gradient example}

Theorem \ref{thm: grad S} provides a practical recipe for computing the gradient of a cost function defined on the manifold of passive linear interferometers. In this section, we provide examples of gradients that will be useful in Section \ref{sec: examples} to find heralded state and gate preparations.

Many state preparation problems with linear optics can be encoded in cost functions that are variants of:
\begin{equation}\label{eq: cost S}
    g(S) = f(\varphi(S)) = \tr{\varphi(S) B  \varphi(S)^\dagger C} \,,
\end{equation}
where $B$ and $C$ are Hermitian matrices and $\varphi(S)$ is the quantum evolution unitary associated with the interferometer $S$. Depending on the setting, $g(S)$ may encode the fidelity or the success probability of a heralded state preparation.

To compute the gradient of \eqref{eq: cost S}, it suffices to compute the gradient of $f(U) = \tr{U B  U^\dagger C}$ and plug it into Theorem \ref{thm: grad S}. To compute the differential \eqref{eq: differential} of $f(U)$ in direction $UX$, we consider a smooth curve $\gamma(t) = U \exp(t X)$:
\begin{equation}
\begin{aligned}
    \Der f(U)[UX] &= \frac{d}{dt}f(U \exp(tX)) \bigg|_{t=0} \\[1mm]
    &= \tr{U X B U^\dagger C + U B (UX)^\dagger C}  \\[1mm]
    &= \tr{XB U^\dagger C U - XU^\dagger C U B} \\[1mm]
    &= \expect{X, -[B, U^\dagger C U]}  \\[1mm] 
    &= \expect{UX, U[U^\dagger C U, B]} \:.
\end{aligned}
\end{equation}
Hence, the Riemannian gradient \eqref{eq: riemannian gradient} of $f$ at $U$ is
\begin{equation}\label{eq: grad cost U}
    \nabla_U f (U) = U[U^\dagger C U, B] \:.
\end{equation}
And, substituting $\nabla_U f$ in Eq. \eqref{eq: grad cost S}, we obtain $\nablag$:
\begin{equation}\label{eq: grad g generic}
    \nablag (S) = S\dphi^\ast\left(\left[\varphi(S)^\dagger C \varphi(S), B\right] \right)\:.
\end{equation}

In Section \ref{sec: gate preparation}, we address the preparation of unitary gates with linear optics. A common cost function for these problems is:
\begin{equation}
    g(S) = \left|\tr{A^\dagger B\varphi(S)C}\right|^2 \:,
\end{equation}
where $A$ is the target matrix (not necessarily unitary) and $B$ and $C$ are complex matrices. Similarly to Eq. \eqref{eq: grad g generic}, its gradient can be computed using basic differentiation rules and Theorem \ref{thm: grad S}:
\begin{equation}\label{eq: grad gate generic}
    \nablag (S) = S\dphi^\ast\left(X-X^\dagger\right)\:,
\end{equation}
where $X = \tr{A^\dagger B\varphi(S)C} \varphi(S)^\dagger B^\dagger A C^\dagger $.

\section{Examples}\label{sec: examples}

In this section, we use Riemannian optimization with the analytical gradients derived in the previous section to tackle two essential problems in passive linear optics: (1) heralded state preparation and (2) heralded gate preparation. In particular, we improved the known success probabilities of some NOON and photon catalysis preparations.

Additionally, in Appendix \ref{appendix: invariants} we give a recipe for faster exact state preparation (states in the same orbit) using Riemannian optimization and Lie algebraic invariants.

\subsection{Heralded state preparation}\label{sec: heralded state prep}

In a heralded state preparation, a target state is prepared probabilistically after measuring a certain heralding pattern in some auxiliary modes. Unlike postselection, heralding does not destroy the target state, so it can be used in further applications \cite{forbes_heralded_2025}.

More precisely, suppose $\inket \in \mathcal{F}^m_n$ is a given input state and we want to prepare a certain target state, $\tket \in \mathcal{F}^{m'}_{n'}$, after measuring a herald, $\hket \in \mathcal{F}^{m-m'}_{n-n'}$, in $m-m'$ auxiliary modes. The heralded preparation problem (Figure \ref{fig: heralded circuit}) consists in finding an interferometer $S$ such that
\begin{equation}
    \outket = \varphi(S) \inket = \sqrt{p} \tket \hket + \sum_i \alpha_i \ket{\psi_i} \ket{ h^i_\perp} \:,
\end{equation}
where $p$ is the success probability of the preparation, each $\ket{\psi_i}\in\mathcal{F}^{m'}_{n'}$ is an arbitrary state and each $\ket{h^i_\perp}\in \mathcal{F}^{m-m'}_{n-n'}$ is orthogonal to the herald.

\begin{figure}[ht]
\centering
\includegraphics[scale=1.8]{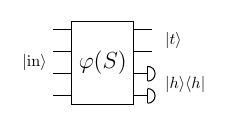}
\caption{Circuit of a heralded state preparation where $\inket$ is the input state, $\tket$ is the target state and $\hket$ is some herald that we wish to measure.}
\label{fig: heralded circuit}
\end{figure}

\begin{table}[t]
\centering
\setlength{\tabcolsep}{6pt}
\renewcommand{\arraystretch}{1.2}
\setlength{\aboverulesep}{0.1ex}
\setlength{\belowrulesep}{0.6ex}
\begin{tabular}{cccc}
\toprule
\addlinespace
\textbf{State prep.} & $\boldsymbol{m}$ & $\boldsymbol{n}$ & $\boldsymbol{\mathcal{P}}$ \\
\addlinespace
\midrule
Bell \cite{stanisic_generating_2017} & 8 & 4 & 18.75\% \\
Bell w/ state injection \cite{fldzhyan_compact_2021}  & 6 & 4 & 7.4\%  \\
NOON, $\ket{40} + \ket{04}$ \cite{lee_noon} & 4 & 6 & 16.53\%$^\ast$ \\
NOON, $\ket{50} + \ket{05}$ & 4 & 6 & 9.168\% \\
GHZ \cite{gubarev_improved_2020}  & 10 & 6  & 0.926\% \\
GHZ linear graph state \cite{bartolucci_creation_2021, hartnett_automated_2026} & 10 & 6 & 1.852\%  \\
GHZ K3 graph state \cite{bartolucci_creation_2021, hartnett_automated_2026} & 12 & 6 & 3.125\%  \\
$\ket{\psi_3}$ w/ photon catalysis  \cite{aralov_photon_2026}  & 3 & 4  & 22.5\%$^\ast$  \\
$\ket{\psi_6} $  w/ photon catalysis  \cite{aralov_photon_2026}  & 2 & 3 & 34.6\%$^\ast$ \\
$\ket{R_2} $  w/ photon catalysis  \cite{aralov_photon_2026}  & 3 & 3  &  100\%$^\ast$ \\
\addlinespace
\bottomrule
\end{tabular}
\caption{Examples of heralded state preparations with different circuit architectures with $m$ modes and $n$ total photons. $\mathcal{P}$ is the highest success probability of the preparation that we could find (for an infidelity lower than $10^{-7}$). The asterisk means that we found a better success probability than the one in the references. The exact input states and heralds, as well as the choices of cost functions and optimizers, can be found in the accompanying repository \cite{optimization_repo}.
}
\label{tab: state preparation examples}
\end{table}

This problem can be turned into an optimization problem by maximizing two figures of merit. One is the probability of measuring the herald,
\begin{equation}\label{eq: heralded prob}
\begin{aligned}    
    \mathcal{P}(\outket, h) &= \|P_{ h} \outket\|_2^2 \\
    &= \tr{\varphi(S) \inket \inbra  \varphi(S)^\dagger P_{h}} \,,
\end{aligned}
\end{equation}
where $P_h$ is the projector onto the subspace of states with heralding pattern $\hket$. The other is the (unnormalized) fidelity of the output with the heralded target, $\thket\equiv\tket\otimes\hket$:
\begin{equation}\label{eq: fidelity}
\begin{aligned}    
    &\mathcal{\tilde F}(\outket, \thket) 
    = |\outbra t, h\rangle|^2 \\[1mm]
    &= \tr{\varphi(S) \inket \inbra  \varphi(S)^\dagger \thket \thbra} \,,
\end{aligned}
\end{equation}
which can be normalized by dividing it by $\mathcal{P}$. As it turns out, both figures of merit are special cases of Eq. \eqref{eq: cost S}. Simply substitute $B=\inket \inbra$ and $C=P_{h} \tket \tbra P_{ h}$ for the fidelity, and $B=\inket \inbra$ and $C = P_{h}$ for the probability.

These two figures of merit can be combined to create a single cost function. A popular choice in the literature \cite{stanisic_generating_2017, gubarev_improved_2020, fldzhyan_compact_2021} is 
\begin{equation}\label{eq: cost gubarev}
    f(S) = -\mathcal{P}^\alpha \mathcal{F}^\beta \:,
\end{equation}
where $ \mathcal{F} = \mathcal{\tilde F}/\mathcal{P}$ is the normalized fidelity and $\alpha$ and $\beta$ are exponents chosen to prioritize a high fidelity over a high success probability. A different cost function is found in \cite{yao_riemannian_2024}, where the authors optimize a piecewise function that first maximizes the fidelity and then maximizes the probability:
\begin{equation}\label{eq: cost piecewise}
    f(S) =
    \begin{cases}
         1 - \mathcal{F}\:, & \text{if }  \mathcal{F} < 0.999 \:, \\ 
        1 - \mathcal{F} - 0.001 \mathcal{P}  \:, & \text{if }  \mathcal{F} \geq 0.999 \:.
    \end{cases}
\end{equation}
Since we already computed the gradient of Eq. \eqref{eq: cost S} in Section \ref{sec: gradient example}, we can apply the sum and product rules of derivatives to compute the gradients of the cost functions \eqref{eq: cost gubarev} and \eqref{eq: cost piecewise}.

In addition to the heralded state preparation scheme with one unitary and one herald (Figure \ref{fig: heralded circuit}), we implemented other architectures with multiple unitaries in our library, such as Adaptive Linear Optics \cite{chabaud_quantum_2021, hoch_quantum_2025}, Single-mode State Injection \cite{state_injection} and Photon Catalysis \cite{aralov_photon_2026}. Computing the derivatives of these circuits is straightforward: one simply has to take the partial derivative with respect to each unitary (recall that any measurement, injection or photon addition between the unitaries is just a linear operator whose differential is itself).

In Table \ref{tab: state preparation examples}, we show examples of heralded preparations obtained with our Riemannian optimizer implemented in QOptCraft \cite{qoptcraft_repo, optimization_repo}. We show a varied selection, including heralded preparations taken from \cite{forbes_heralded_2025} and references therein, and photon catalysis preparations from \cite{aralov_photon_2026}. The reported probabilities in all cases are either equal to the known ones or higher (marked with an asterisk).

For example, we found a better success probability for the generation of the NOON state $\ket{40} + \ket{04}$ than the one reported in \cite{lee_noon} (for the same input and herald), a result that addresses their open question of whether their setup was optimal. Additionally, for the generation of the NOON state $\ket{50} + \ket{05}$, we found a success probability of 9.168\% using $\ket{2,2,1,1}$ as input, which improves on the 5.64\% probability found by VanMeter et al. \cite{vanmeter_general_2007} using $\ket{2,2,2,0}$ as input. As an example, we show in Figure \ref{fig: noon circuit} the circuit for the $\ket{50} + \ket{05}$ preparation.

\begin{figure}[ht]
\centering
\includegraphics[scale=0.28]{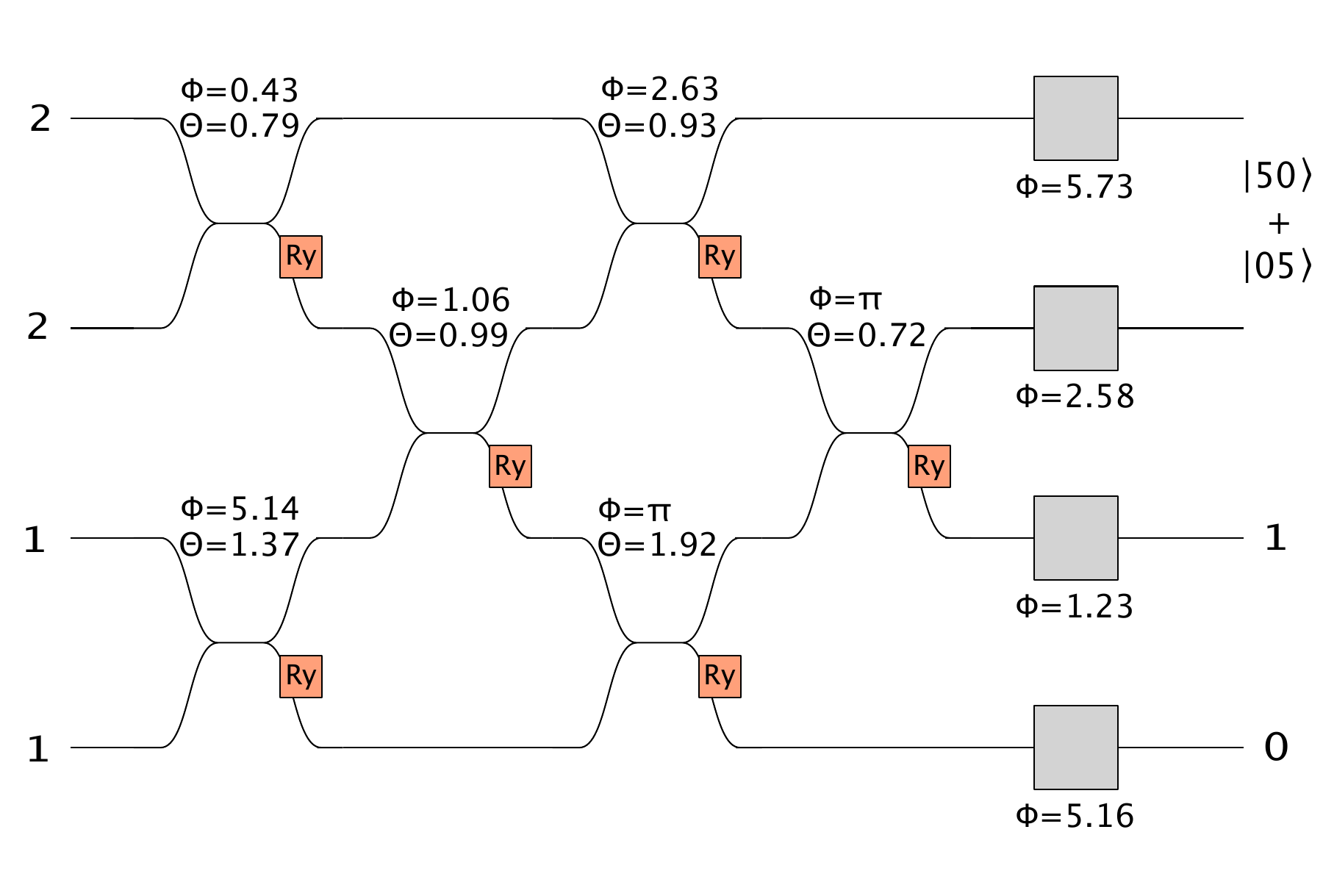}
\caption{Passive circuit for preparing the NOON state $(\ket{50}+\ket{05})/\sqrt{2}$ using $\ket{2,2,1,1}$ as input and $\ket{1,0}$ as herald. The preparation succeeds with probability 9.168\%, which improves on the 5.64\% found by VanMeter et al. \cite{vanmeter_general_2007} for the same NOON. The circuit diagram was drawn using Perceval \cite{heurtel_perceval_2023}.}
\label{fig: noon circuit}
\end{figure}

Regarding the photon catalysis preparations, our optimizer improved some success probabilities from Table IV of \cite{aralov_photon_2026}. Some examples are shown in Table \ref{tab: state preparation examples} with an asterisk: for $\ket{\psi_3}$ we found $22.5\%$ instead of 19\%; for $\ket{\psi_6}$, 34.6\% instead of 32\%; and, for $\ket{R_2}$, 100\% instead of 82\%. Additional improvements for other states can be found in the accompanying notebooks \cite{optimization_repo}.

\subsection{Heralded gates}\label{sec: gate preparation}

We now switch to a related problem: heralded gate preparation with linear optics, which is useful for optical quantum computing and for preparing graph states \cite{kok_linear_2007}.

To prepare an $N_q$-qubit quantum gate with linear optics, the first step is to define a qubit basis, $\{\ket{c_j}: 1\leq j \leq 2^{N_q} \}$. A common choice is the dual-rail encoding, where each qubit is represented by a pair of modes and one photon \cite{kok_linear_2007}: $\ket{0}_L = \ket{1,0}$ and $\ket{1}_L = \ket{0,1}$. Two-qubit gates like the CNOT cannot be implemented deterministically with linear optics \cite{knill_scheme_2001, knill_quantum_2002, carolan_universal_2015}. Any preparation relies on having auxiliary photons in the input, $\ket{a}$, and measuring a heralding pattern, $\hket$, in the auxiliary modes. In \cite{uskov_maximal_2009}, the heralded gate preparation is modeled by restricting the input states of the multi-photon unitary, $\varphi(S)$, to the computational basis plus the auxiliary photons, $\{\ket{c_j}\ket{a}\}$. Meanwhile, the output states are restricted to the subspace of Fock states with the desired heralding pattern, $\{\ket{k_j}\hket\}$. Mathematically, the restricted multi-photon unitary is given by
\begin{equation}
    \mathrm{A}(S) = V_{\text{out}}^\dagger \varphi(S) V_{\text{in}} \:,
\end{equation}
where $V_{\text{in}} = \sum_{j}
\ket{c_j, a}
\bra{c_j}$ and 
$V_{\text{out}} = \sum_{j}\ket{k_j, h}\bra{k_j}$ are isometries (matrices such that $V^\dagger V = \mathrm{Id}$). The elements $\bra{c_j}\mathrm{A}(S)\ket{c_i}$ implement the quantum gate, while the elements $\bra{k_j}\mathrm{A}(S)\ket{c_i}$, for $\ket{k_j}\neq \ket{c_j}$, are undesired transitions.

To find the $A(S)$ that best approximates the target gate, $A_t$, Uskov et al. \cite{uskov_maximal_2009} optimize a piecewise function \eqref{eq: cost piecewise} that first maximizes the gate fidelity,
\begin{equation}\label{eq: gate fidelity}
    \mathcal{F}(S) = \frac{\expect{A, A_t} \expect{A_t, A}}{\expect{A, A}\expect{A_t, A_t}}  \:,  
\end{equation}
where $\expect{A, B} = \mathrm{Tr}(A^\dagger B)$ is the Hilbert-Schmidt inner product, and then maximizes the probability of measuring the herald,
\begin{equation}\label{eq: gate success}
    \mathcal{P}(S) = \frac{\expect{A, A}}{2^{N_q}}\:.
\end{equation}

While Uskov et al. optimize a parametrization of the algebra, $S=\exp(\sum_j G_j)$, we can apply Riemannian optimization directly to the cost functions \eqref{eq: gate fidelity}–\eqref{eq: gate success}. The gradient of $|\expect{A, A_t}|^2$ can be computed from Eq. \eqref{eq: grad gate generic} by identifying $A\equiv A_t$, $B\equiv  V_{\text{out}}^\dagger$ and $C\equiv V_{\text{in}}$. The gradient of $\expect{A, A}$ can be computed similarly to the gradient of Eq. \eqref{eq: cost S}. Once we have the gradient of the fidelity and the success probability, they can be fed to the Riemannian optimizer to find heralded gate preparations.

\begin{table}[H]
\centering
\setlength{\tabcolsep}{6pt}
\renewcommand{\arraystretch}{1.2}
\setlength{\aboverulesep}{0.1ex}
\setlength{\belowrulesep}{0.6ex}
\begin{tabular}{cccc}
\toprule
\addlinespace
\textbf{Gate prep.}  & \shortstack{\textbf{Mean time}\\\textbf{(conv.) [s]}} & \shortstack{\textbf{Mean iters}\\\textbf{(conv.)}}  & $\boldsymbol{\mathcal{P}}$  \\
\addlinespace
\midrule
CZ & 0.004 & 44  & 7.4\% \\
CCZ & 0.3 & 460 & 0.34\% \\
\bottomrule
\end{tabular}
\caption{Results for the optimization of a heralded CZ gate (6 modes, 4 photons) and a heralded CCZ gate (12 modes, 6 photons) using our Riemannian BFGS optimizer. $\mathcal{P}$ is the highest success probability of the preparation that we could find (for an infidelity lower than $10^{-7}$), which agrees with the values found in \cite{uskov_maximal_2009}.
}
\label{tab: gate preparation examples}
\end{table}

In Table \ref{tab: gate preparation examples}, we show the results of the optimization of a heralded CZ gate (with 2 auxiliary modes and 2 auxiliary photons) and a heralded CCZ gate (with 6 auxiliary modes and 3 auxiliary photons), both in dual-rail encoding. As we can see, in both cases the Riemannian BFGS optimizer reached the best-known probabilities with a small number of iterations in fractions of a second.

\section{Implementation and benchmarks}\label{sec: benchmarks}

\begin{table*}[t]
\centering
\setlength{\tabcolsep}{6pt}
\renewcommand{\arraystretch}{1.2}
\setlength{\aboverulesep}{0.1ex}
\setlength{\belowrulesep}{0.6ex}
\begin{tabular}{ccccccc}
\toprule
\addlinespace
\textbf{State prep.} & \textbf{Modes} & \textbf{Photons}  & \textbf{Optimizer} & \shortstack{\textbf{Mean time}\\\textbf{(conv.) [s]}}  &  \shortstack{\textbf{Mean iters}\\\textbf{(conv.)}} &\shortstack{ \textbf{Optimality}\\\textbf{rate}}   \\
\addlinespace
\midrule
Bell & 6 & 4 & Riemannian BFGS & 0.012 & 110 & 75\% \\
Bell & 6 & 4 & BFGS Clements & 0.75 & 310 & 65\% \\
Bell & 6 & 4 & BFGS algebra & 0.30 & 280 & 73\%  \\
Bell & 6 & 4 & Piquasso & 0.31 & 400 & 63\%  \\
Bell & 6 & 4 & MrMustard & 70  & 2900 &  40\% \\
\addlinespace
NOON & 4 & 6 & Riemannian BFGS & 0.006 & 56  &  65\% \\
NOON & 4 & 6  & BFGS Clements & 0.093 & 99 & 38\% \\
NOON  & 4 & 6  & BFGS algebra & 0.068  & 110 & 56\% \\
NOON  & 4 & 6  & Piquasso & 0.12 & 210 & 42\%  \\
NOON & 4 & 6 & MrMustard & 2.5 & 370 & 4\% \\
\addlinespace
GHZ  & 10 & 6  & Riemannian BFGS & 0.24 & 510 & 12\% \\
GHZ & 10 & 6 & BFGS Clements & 30 & 2100 & 0.6\% \\
GHZ & 10 & 6 & BFGS algebra & 5.0 & 1600 & 11\%  \\
GHZ  & 10 & 6 & Piquasso & 1200 & 370 & 0.2\% \\
GHZ  & 10 & 6  & MrMustard  & --- & --- & --- \\

\bottomrule
\end{tabular}
\caption{Benchmarking of heralded state preparations using several optimizers. We optimized the cost function \eqref{eq: cost gubarev}, with $\alpha = 10^{-3}$ and $\beta = 4$, considering the optimization successful when an infidelity of $10^{-7}$ was reached and the best-known success probabilities \cite{forbes_heralded_2025} were found: $7.4\%$ or $7.78\%$ for the Bell preparation \cite{gubarev_fock_2021, fldzhyan_compact_2021}; $5.64\%$ for the NOON \cite{vanmeter_general_2007}, and $0.925\%$ for the GHZ \cite{gubarev_improved_2020} (it doubles to $1.852\%$ when considering an additional heralding pattern). We analyzed three performance metrics: \textit{optimality rate} (percentage of Haar-random initial unitaries that converge to an optimal solution), wall-clock time and mean iterations to converge to an optimal solution (for those runs that reach the optimum). We left the results for the GHZ optimization with MrMustard empty because the computer ran out of RAM. The benchmarks were performed using QOptCraft v2.8.1, the BFGS in SciPy v1.17.1 \cite{Scipy}, Piquasso v7.2.2 \cite{kolarovszki_piquasso_2025} and MrMustard v1.0.0a1 \cite{yao_riemannian_2024} on a workstation equipped with an Apple M4 Max and 128 GB DDR5 RAM.
}
\label{tab: optimizer benchmarks}
\end{table*}

In this section, we review some of the key points in our implementation in QOptCraft of Riemannian optimization applied to heralded state preparation and benchmark it against other libraries and optimization methods. 

\subsection{QOptCraft implementation}\label{sec: qoptcraft implementation}

The main bottleneck in computing cost functions like the fidelity, Eq. \eqref{eq: fidelity}, lies in computing $\varphi(S)$. However, in most state preparation problems, the input, target and herald states are simple enough to be written as a superposition of a few number states $|\vec{n}\rangle$. In this case, to compute the fidelity, it is enough to compute $\varphi(S)|\vec{n}\rangle$ for every number state $|\vec{n}\rangle$ in $\inket$. Moreover, since we only need the amplitudes of output states with a certain heralding pattern, we can use Algorithm 2 (\texttt{SLOS\_gen}) of \cite{heurtel_strong_2023} to avoid computing the full output state. This differs from libraries like Piquasso \cite{kolarovszki_piquasso_2025}, whose simulator computes all the output amplitudes, resulting in a slower optimization. 

Computing the gradient of the fidelity (recall Eq. \eqref{eq: grad g generic}) requires computing
\begin{equation}\label{eq: fast grad fid}
    \varphi(S)^\dag \nabla_U f = \Big[\varphi(S)^\dag \thket \thbra \varphi(S) , \inket \inbra \Big]
\end{equation}
where $\thket \equiv \ket{t} \otimes \hket$. Thus, we also need to compute $\varphi(S)^\dag|\vec{n}\rangle$ for every number state in $\thket$. With this strategy, we compute the elements of $\varphi(S)$ that we actually need, speeding up the optimization. Note that, to compute the gradient of the success probability in Eq. \eqref{eq: heralded prob}, we also need to compute $\varphi(S)^\dag|\vec{n}\rangle$ for every number state with heralding pattern $\hket$.

Another expensive step in the optimization is storing the gradient $\nabla_U f$, of size $D\times D$. However, we can exploit the projector structure of gradients like Eq. \eqref{eq: fast grad fid} to directly compute $\dphi^\ast\left(\varphi(S)^\dag \nabla_U f \right)$, of much smaller size $m\times m$. The only large matrices that appear are the $\dphi(G_i)$ in the adjoint map \eqref{eq: adjoint algebra homomorphism}, but they are very sparse matrices that can be stored efficiently. 

These optimizations, which we illustrated for the heralded state preparation, carry over analogously to the heralded gate preparation from Section \ref{sec: gate preparation} and other linear optical problems.

\subsection{Benchmarks}

In this section, we compare our Riemannian optimizer against different optimization approaches and implementations found in the literature. In particular, we benchmark the following optimizers:
\begin{enumerate}
    \item Our Riemannian BFGS optimizer using analytical gradients, implemented in QOptCraft \cite{qoptcraft_repo}.
    \item A custom Euclidean BFGS optimizer using numerical gradients. Here we compare two parametrizations of the scattering matrix: one given by the Clements decomposition \cite{clements_optimal_2016} (a similar method to \cite{fldzhyan_compact_2021}) and another given by the parametrization of the algebra plus the Cayley transform (similar method to \cite{gubarev_improved_2020}).
    \item The Riemannian ADAM optimizer implemented in MrMustard \cite{yao_riemannian_2024}.
    \item A Euclidean BFGS optimizer using the analytical gradients computed with Piquasso's simulator \cite{kolarovszki_piquasso_2025} (which uses automatic differentiation).
\end{enumerate}

In Table \ref{tab: optimizer benchmarks}, we show the benchmarks for the optimization of three typical heralded state preparations: dual-rail Bell \cite{carolan_universal_2015, gubarev_improved_2020, fldzhyan_compact_2021}, NOON \cite{vanmeter_general_2007} and GHZ \cite{gubarev_improved_2020}. As we can see, our Riemannian BFGS is several orders of magnitude faster than the other optimizers. It also requires fewer iterations (less than half), which matters because every iteration requires computing many permanents. Lastly, we show the percentage of Haar-random initial unitaries that converge to an optimal solution (\textit{optimality rate}), which provides a rough idea of the number of local minima in the optimization landscape. Interestingly, the optimality rate for the NOON is much higher with Riemannian optimization and it is also among the highest for the Bell and GHZ.

As we can see, the Euclidean optimization of the parameters in the Clements decomposition \cite{clements_optimal_2016} with numerical gradients (the approach followed in \cite{fldzhyan_compact_2021}) is relatively fast and accurate for optimizations with a few modes and photons. However, it becomes slower for the GHZ optimization and the percentage of random seeds that converge drops sharply, possibly due to numerical errors caused by finite differences. In contrast, the Euclidean optimization of the algebra parametrization (the approach followed in \cite{gubarev_improved_2020, hartnett_automated_2026}) remains within a few seconds and maintains a high optimality rate even for the GHZ, although still lower than Riemannian BFGS.

The benchmarks also show the Euclidean optimization of the Clements decomposition modeled with Piquasso \cite{kolarovszki_piquasso_2025} and SciPy's BFGS implementation \cite{Scipy}. The advantage is that Piquasso has automatic differentiation capabilities, so we can feed an analytical gradient to the BFGS optimizer. This method proved to be relatively fast for problems with few modes and photons. However, it becomes very slow and has a low optimality rate for larger problems. One reason could be that Piquasso's simulator computes all the amplitudes of the evolved state before heralding, resulting in many unnecessary computations that quickly raise the computational complexity. In contrast, our implementation in QOptCraft uses the strong simulation (\texttt{slos\_gen}) algorithm from \cite{heurtel_strong_2023} to compute only the amplitudes with the desired heralding pattern, greatly reducing the number of Fock state evolutions that have to be computed.

Finally, we see that MrMustard \cite{yao_riemannian_2024}, even though it also uses Riemannian optimization, is the slowest of the methods and has a much lower optimality rate. One explanation could be that MrMustard uses a first-order optimization method (AdaBelief by default), while QOptCraft uses a second-order Riemannian BFGS (in Appendix \ref{appendix: optimizer comparison} we show the same trend for a different cost function). Moreover, MrMustard is designed with a focus on Gaussian state simulations, rather than Fock state simulations, and uses a cutoff for the number of photons that led the computer to run out of RAM (about 128 GB) during the GHZ heralded preparation.

In conclusion, QOptCraft's Riemannian BFGS optimizer (which uses the analytical gradients derived in Section \ref{sec: gradient example} and the efficient implementation from Section \ref{sec: qoptcraft implementation}) is orders of magnitude faster than other state-of-the-art optimizers found in the literature, while keeping a nice optimization landscape (i.e. a high rate of random seeds converging to an optimum). This shows the advantage of Riemannian optimization in larger-scale problems.

\section{Discussion}\label{sec: discussion}

In this work, we implemented a new framework for the Riemannian optimization of passive linear optical problems. Our optimizer uses a closed-form analytical gradient of the cost function and optimizes the unitary matrix of the photonic circuit without relying on any parametrization (e.g., the Clements decomposition \cite{clements_optimal_2016}). As shown in Section \ref{sec: benchmarks}, removing the Clements parametrization improves the optimization landscape, as can be seen from the higher rate of optimal results and the lower number of iterations required to reach a minimum. Our benchmarks also show a speedup of several orders of magnitude, which opens the door for optimizing over larger Hilbert spaces (of states with more modes and photons).

One disadvantage of not parametrizing the unitary is that it is not well suited for directly reducing the number of gates in the circuit, which would require optimizing for gate parameters equal to zero \cite{fldzhyan_compact_2021}. However, once an optimal unitary is found, a second-stage optimization and compilation could be run to reduce the optical depth and the number of optical components \cite{hartnett_automated_2026}, obtaining an equivalent interferometer with fewer experimental losses.

Our approach also differs from the Riemannian optimizer implemented in MrMustard \cite{yao_riemannian_2024} in three aspects: (1) we optimize passive circuits instead of Gaussian circuits; (2) we give exact formulas for the gradients rather than relying on automatic differentiation; (3) in addition to first-order methods, like Riemannian gradient descent or ADAM, we implemented second-order quasi-Newton methods, like Riemannian BFGS and L-BFGS, which give faster convergence.

We did not compare our optimizer with Hartnett et al. \cite{hartnett_automated_2026} since the code is not available, but there are several differences with our method. First, we compute the evolution of multiphoton states using the strong simulation method from Heurtel et al. \cite{heurtel_strong_2023}, while Hartnett et al. compute them using a Fast Fourier Transform-based simulation method. FFT-based simulation has a worse computational complexity than strong simulation, but it provides two advantages: (1) it is compatible with automatic differentiation (which is not a problem in our method since we use a closed analytical formula for the gradient); (2) FFT-based simulation can be GPU-accelerated, which enables them to search for graph states with up to 5 qubits. This contrasts with our implementation of Riemannian optimization, which runs exclusively on CPU.

There are many applications of linear optical optimization in the literature that could benefit from Riemannian optimization. Many open questions remain in heralded state preparation \cite{forbes_heralded_2025}, such as finding optimal preparations for graph states \cite{hartnett_automated_2026}, which are required for measurement-based linear optical quantum computation \cite{kok_linear_2007}. However, optimization in linear optics goes far beyond state preparation. Examples of other recent applications include verifying theorems \cite{aralov_photon_2026}, testing conjectures \cite{francalanci_local_2026} and studying entanglement \cite{lachman_certification_2026}. We leave it for future work to test Riemannian optimization in these use cases.

Finally, it is worth noting that optimizing over the unitary group is a very complicated problem. The group is not convex and the cost functions are full of local minima, so there is no hope of turning the optimization into an easy problem. However, any improvement in the cost function landscape and the optimization speed is essential for pushing the current boundaries of passive linear optical optimization.
\vspace{2.2cm}

\section{Data availability}

The code developed for this paper is available in the open-source library QOptCraft \cite{aguado_qoptcraft_2023, qoptcraft_repo} and the examples and benchmarks can be found in the accompanying repository \cite{optimization_repo}.

\section{Acknowledgements}

The author would like to thank J. C. García-Escartín, V. Gimeno i Garcia, J. J. Moyano-Fernández, U. Chabaud, C. Lopetegui, M. Frigerio and S. Fldzhyan for stimulating discussions and useful comments. P.~V.~P. has been funded under the UVa 2024 predoctoral contract, co-financed by Banco Santander; the Q-CAYLE project, funded by the European Union-Next Generation UE/MICIU/Plan de Recuperación, Transformación y Resiliencia/Junta de Castilla y León (PRTRC17.11); and the Department of Education, Junta de Castilla y León and FEDER Funds (Reference: CLU-2023-1-05).

\clearpage

\appendix

\section{Other optimization algorithms}\label{appendix: algorithms}

Riemannian gradient descent only guarantees linear convergence, which can be slow for large problems. This has motivated the proposal of Riemannian optimization algorithms that exploit second-order derivatives or add momentum to increase the rate of convergence.

\subsection{Riemannian BFGS}

Newton's optimization method improves gradient descent \eqref{eq: eucl grad descent} by replacing the gradient, $\nabla f(x_k)$, with a new descent direction, $\eta_k =  \nabla^{-2} f(x_k) \nabla f(x_k)$, where $\nabla^{-2} f$ is the inverse of the Hessian of $f$. In this way, by using second-order derivatives, Newton's method achieves quadratic convergence instead of linear convergence.

However, computing and storing the inverse Hessian can be expensive. This has motivated the introduction of quasi-Newton methods, which sacrifice the fast convergence rate in exchange for a more efficient algorithm. The key idea behind them is to compute approximations of the inverse Hessian, $B_k$, with finite differences using the so-called \textit{secant condition}:
\begin{equation}\label{eq: secant condition} x_{k+1} - x_k = B_{k+1} \left(\nabla f(x_{k+1}) - \nabla f(x_k) \right)\:.  
\end{equation}
Two well-known quasi-Newton methods are BFGS (Broyden–Fletcher–Goldfarb–Shanno) and L-BFGS, its limited-memory version, both of which achieve superlinear convergence \cite{Dennis1996}.

The same computational efficiency argument can be applied to the Riemannian Newton's method, so several versions of Riemannian BFGS \cite{Gabay1982, huang_bfgs, rw12, hga15} and L-BFGS \cite{hga15} have been proposed. However, the secant condition \eqref{eq: secant condition} is not well defined on a Riemannian manifold. Firstly, $x_{k+1}$ and $x_k$ are points on a manifold, so their difference is not defined. Secondly, the gradients $\nabla f(x_{k+1})$ and $\nabla f(x_k)$ live in different tangent spaces, so their difference is not defined either. To write a Riemannian secant condition, one first needs to map tangent vectors from $T_{x_k}\mathcal{M}$ to $T_{x_{k+1}}\mathcal{M}$ using a tool called a \textit{transporter}.

\begin{definition}[Definition 10.61, \cite{boumal_introduction_2023}]
Given a manifold $\mathcal{M}$, let $\mathcal{V}$ be open in $\mathcal{M} \times \mathcal{M}$ such that $(x, x) \in V$ for all $x \in \mathcal{M}$.
A transporter on $V$ is a smooth map
\begin{equation*}
    \begin{aligned}
        \mathcal{T} : \:\: V \: & \to \mathcal{L}(T\mathcal{M}, T\mathcal{M})\\
    (x, y) &\mapsto \mathcal{T}_{y \leftarrow x}
    \end{aligned}    
\end{equation*}
such that $\mathcal{T}_{y \leftarrow x}$ is linear from $\mathcal{T}_x\mathcal{M}$ to $\mathcal{T}_y\mathcal{M}$ and $\mathcal{T}_{x \leftarrow x}$ is the identity.
\end{definition}

A simple example of an (isometric) transporter on the unitary group is
\begin{equation}\label{eq: trivial transporter}
    \mathcal{T}_{U_{k+1} \leftarrow U_k}(U_k X) = U_{k+1} X  \:.
\end{equation}
An important technical property of this transporter is that $\mathcal{T}_{R_{U}(X) \leftarrow U}(U X) = \frac{d}{dt} R_U(X+tX)|_{t=0}=U\exp(X)X$ when $R_U(X)$ is the exponential retraction \eqref{eq: exp retraction}. This so-called \textit{locking condition} allows us to apply the Riemannian BFGS algorithm in \cite{hga15} to the unitary group and explains why we only use the exponential retraction with BFGS. 

Returning to the secant condition \eqref{eq: secant condition}, the transporter is used to substitute $x_{k+1} - x_k$ by
\begin{equation}
    s_k = \mathcal{T}_{x_{k+1} \leftarrow x_k}(\alpha_k\eta_k)  \:,
\end{equation}
where $\eta_k$ is the descent direction and $\alpha_k$ is the step size. Analogously, 
\begin{equation}
    y_k = \nabla f(x_{k+1}) - \mathcal{T}_{x_{k+1} \leftarrow x_k}(\nabla f(x_k))
\end{equation}
is the manifold version of $\nabla f(x_{k+1}) - \nabla f(x_k)$. Thus, the Riemannian analogue of the secant condition \eqref{eq: secant condition} is $s_k = B_{k+1} y_k$.

In the literature, several formulas for the inverse Hessian satisfying the secant condition have been proposed. Our implementation follows the BFGS update from \cite{hga15}:
\begin{equation} \label{eq: hessian bfgs update}
    \begin{aligned}
    B_{k+1} v = \tilde{B}_k v &+ \frac{\langle s_k, v \rangle}{\langle y_k, s_k \rangle} s_k - \frac{\langle s_k, v \rangle}{\langle y_k, s_k \rangle} \tilde{B}_k y_k \\
    &- \frac{\langle y_k, \tilde{B}_k v \rangle}{\langle y_k, s_k \rangle} s_k + \frac{\langle y_k, \tilde{B}_k y_k \rangle \langle s_k, v \rangle}{\langle y_k, s_k \rangle^2} s_k \:.
    \end{aligned}
\end{equation}
The transformation $\tilde{B}_k = \mathcal{T}_{x_{k+1} \leftarrow x_k} \circ B_k \circ \mathcal{T}_{x_k \leftarrow x_{k+1}}$ is necessary because $B_k$ acts on vectors in $T_{x_k}\mathcal{M}$ while $y_k$ and $s_k$ are vectors in $T_{x_{k+1}}\mathcal{M}$. In practice, the matrix of $B_{k+1}$ is computed from the coordinates of $s_k$ and $y_k$ in the tangent space, which are easy to compute for our choice of transporter \eqref{eq: trivial transporter}.

Last but not least, BFGS (Algorithm \ref{alg: bfgs}) needs a line search that satisfies the weak Wolfe conditions (Algorithm \ref{alg: wolfe line search}), which consist of the Armijo-Goldstein condition and a curvature condition. Note that, to check the curvature condition, we need to compute $\frac{d}{dt}f(\mathrm{R}_x(t\eta))|_{t=\alpha}$, which is easy when $R$ is the exponential retraction.

\begin{algorithm}[H]
\caption{Riemannian BFGS \cite{hga15}}\label{alg: bfgs}
\begin{algorithmic}[1]
\Require Initial inverse Hessian, $B_0 = \mathrm{Id}$.
\State \textbf{Input:} $x_0 \in \mathcal{M}$
\For{$k = 0, 1, 2, \ldots$}
    \State $\eta_k = -B_k \nabla f(x_k)$
    \State Pick $\alpha_k > 0$ with the Wolfe line-search (Algorithm \ref{alg: wolfe line search})
    \State $x_{k+1} = \mathrm{R}_{x_k}(\alpha_k\eta_k)$
    \State $s_k = \mathcal{T}_{x_{k+1} \leftarrow x_k}(\alpha_k\eta_k)$
    \State $y_k = \nabla f(x_{k+1}) - \mathcal{T}_{x_{k+1} \leftarrow x_k}(\nabla f(x_k))$\vspace{0.5mm}
    \State Compute $B_{k+1}$ with Eq. \eqref{eq: hessian bfgs update}
\EndFor
\State \textbf{Output:} Local minimum of $f$
\end{algorithmic}
\end{algorithm}

\begin{algorithm}[H]
\caption{Wolfe line-search \cite{nocedal_numerical, hga15}}\label{alg: wolfe line search}
\begin{algorithmic}[1]
\Require Constants $c_1 \in (0,\frac{1}{2})$, $c_2 \in (\frac{1}{2},1)$ (for example, $c_1 = 10^{-4}$ and $c_2 = 0.999$), $\tau_1 \in (0,1)$, $\tau_2 > 1$.
\State \textbf{Input:} $x \in \mathcal{M}$, search direction $\eta \in T_x\mathcal{M}$, $\bar{\alpha} > 0$.
\State Set $\alpha \leftarrow \bar{\alpha}$
\While{Wolfe conditions not satisfied}
    \If{$f(\mathrm{R}_x(\alpha \eta)) > f(x) + c_1 \alpha g(\operatorname{grad} f(x), \eta)$}\vspace{0.5mm}
        \State Decrease $\alpha$: $\alpha \leftarrow \tau_1 \alpha$ \vspace{0.5mm}
    \ElsIf{$\frac{d}{dt}f(\mathrm{R}_x(t\eta))|_{t=\alpha} < c_2 \frac{d}{dt}f(\mathrm{R}_x(t\eta))|_{t=0}$}\vspace{0.5mm}
        \State Increase $\alpha$: $\alpha \leftarrow \tau_2 \alpha$ \vspace{0.5mm}
    \Else
        \State \textbf{break} \Comment{Wolfe conditions are satisfied}
    \EndIf
\EndWhile
\State \textbf{Output:} $\alpha$
\end{algorithmic}
\end{algorithm}

In addition to Riemannian BFGS, we implemented its limited-memory version, Riemannian L-BFGS, which does not explicitly compute $B_k$. Instead, it stores $\ell$ vectors $s_k$ and $y_k$ from previous iterations and uses them to implicitly approximate $B_k$ and compute the new descent direction. Since the algorithm is more involved, we refer the reader to Algorithm 2 of \cite{hga15}.

\subsection{Gradient descent with momentum}

The slow convergence of gradient descent motivated the introduction of Polyak’s heavy-ball method \cite{POLYAK1964}, in which the gradient accumulates ``momentum'' from previous iterations, like a ball rolling down a hill, allowing the optimizer to overshoot some local minima in the hope of finding a lower one. The method was later refined by Nesterov’s accelerated gradient method \cite{nesterov_introductory_2004} and, nowadays, momentum algorithms like ADAM \cite{kingma2015adam} have become very popular for training neural networks.

\begin{algorithm}[H]
\caption{Momentum optimizer \cite{kong_quantitative_2024}}\label{alg: momentum}
\begin{algorithmic}[1]
\Require Lie group $\mathcal{G}$, step size $h > 0$, friction $\gamma > 0$
\State \textbf{Input:} $x_0 \in \mathcal{G}$,\ $\eta_0 = 0$,\ $\nabla f_0 = 0$
\For{$k = 0, 1, 2\ldots$}
    \If{\textit{Heavy-Ball}}
        \State $\eta_{k+1} = (1 - \gamma h)\,\eta_k - h\,x_k^{-1}\nabla f(x_k)$
    \EndIf
    \If{\textit{Nesterov accelerated gradient}}
        \State $\alpha \leftarrow 1 - \gamma h$
        \State $\eta_{k+1} = \alpha\,\eta_k 
            - \alpha\, h \bigl(x_k^{-1}\nabla f(x_k) - x_{k-1}^{-1}\nabla f(x_{k-1})\bigr)$
        \Statex    \hspace{2cm} $- h\,x_k^{-1}\nabla f(x_k)$
    \EndIf
    \State $x_{k+1} = \mathrm{R}_{x_k}(h\,\eta_{k+1})$
\EndFor
\State \textbf{Output:} Local minimum of $f$
\end{algorithmic}
\end{algorithm}

With the increasing popularity of Riemannian optimization, momentum methods have also been introduced \cite{riemannian_adam}, in particular in the context of optimizing over Lie groups \cite{tao_variational_2020, kong_quantitative_2024, campos_momentum_2024}. In our library, QOptCraft \cite{aguado_qoptcraft_2023}, we implemented the \textit{heavy-ball} and the Nesterov accelerated gradient methods following \cite{kong_quantitative_2024} (Algorithm \ref{alg: momentum}), and a Riemannian ADAM following \cite{riemannian_adam} (Algorithm \ref{alg: radam}).

\begin{algorithm}[H]
\caption{Riemannian ADAM \cite{riemannian_adam}}
\label{alg: radam}
\begin{algorithmic}[1]
\Require Manifold $\mathcal{M}$, step sizes ${\alpha_k}$, exponential decay rates for the moment estimates $\beta_1$, $\beta_2$ in the interval $[0,1)$, machine epsilon $\epsilon$.

\State \textbf{Input:} $x_0 \in \mathcal{M}$, $\tau_0 = 0$, $v_0 = 0$.
\For{$k = 0, 1, 2, \ldots$}
    \State $g_k \leftarrow \mathrm{grad}\,f(x_k)$
    \State $m_{k+1} \leftarrow 
        \beta_1\,\tau_k
        + (1 - \beta_1)\,g_k$
    \State $\hat m_{k+1} \leftarrow m_{k+1} / (1 - \beta_1^{k+1})$
    \State $v_{k+1} \leftarrow
        \beta_2\,v_k
        + (1 - \beta_2)\,\|g_k\|_{x_k}^{2}$
    \State $\hat v_{k+1} \leftarrow v_{k+1} / (1 - \beta_2^{k+1})$
    \State $x_{k+1} \leftarrow \mathrm{R}_{x_k}(
                -\,\alpha\, \hat m_{k+1} / (\sqrt{\hat v_{k+1}} +\epsilon))$
    \State $\tau_{k+1} \leftarrow
    \mathcal{T}_{x_{k+1} \leftarrow x_k}(m_{k+1})$
\EndFor
\State \textbf{Output:} Local minimum of $f$
\end{algorithmic}
\end{algorithm}

\subsection{Optimizer comparison}\label{appendix: optimizer comparison}

In Table \ref{tab: qoptcraft optimizer comparison}, we compare several Riemannian optimization algorithms.

\begin{table}[H]
\centering
\setlength{\tabcolsep}{6pt}
\renewcommand{\arraystretch}{1.2}
\setlength{\aboverulesep}{0.1ex}
\setlength{\belowrulesep}{0.6ex}
\begin{tabular}{ccccc}
\toprule
\addlinespace
\textbf{Problem} & \textbf{Optimizer} & \shortstack{\textbf{Mean}\\\textbf{time [s]}}  & \shortstack{\textbf{Mean}\\\textbf{iters}} & \shortstack{\textbf{Optimality}\\\textbf{rate}} \\
\addlinespace
\midrule
Bell prep. & GD  & 0.13   & 930 & 9\% \\
Bell prep. & ADAM & 0.65 & 9900 & 13\%  \\
Bell prep. & Nesterov & 2.7 & 40000 & 16\%  \\
Bell prep. & BFGS & 0.013 & 150 & 78\%  \\
Bell prep. & L-BFGS & 0.034 & 290 & 78\% \\
\addlinespace
NOON prep. & GD  & 0.08  & 600 & 11\% \\
NOON prep. & ADAM & 1.3 & 22000 & 50\%  \\
NOON prep. & Nesterov & 1.7 & 28000 & 25\%  \\
NOON prep. & BFGS & 0.0046 & 61 & 63\%  \\
NOON prep. & L-BFGS & 0.0062 & 78 & 62\% \\
\addlinespace
GHZ prep. & GD  & --- &  --- & 0\%  \\
GHZ prep. & ADAM & --- &  --- & 0\% \\
GHZ prep. & Nesterov & --- & --- & 0\%   \\
GHZ prep. & BFGS & 0.085  & 360 & 1.8\% \\
GHZ prep. & L-BFGS & 0.17 & 890 & 1.5\%   \\
\addlinespace
\bottomrule
\end{tabular}
\caption{Benchmarking of heralded state preparations using several Riemannian optimizers implemented in QOptCraft with the cost function \eqref{eq: cost piecewise}. The metrics used are the same as in Table \ref{tab: optimizer benchmarks}. The exact hyperparameters used can be found in the code repository \cite{optimization_repo}. We do not report results for the first-order optimizers in the GHZ preparation because they did not converge after 50000 steps.
}
\label{tab: qoptcraft optimizer comparison}
\end{table}

We found that quasi-Newton methods outperform first-order methods in both speed and optimality rate (percentage of Haar-random initial unitaries that converge to a known optimal solution). In fact, first-order optimizers failed to converge after 50000 iterations for the GHZ preparation, making them unreliable for larger optimizations.

\section{Faster exact state preparation with invariants}\label{appendix: invariants}

A simpler state preparation problem is one where there is no herald and we want to find the unitary that exactly maps the input state into the target state. This amounts to asking whether the two states are in the same orbit under the action of linear optical unitaries.

We already know from no-go theorems \cite{parellada_no-go_2023, singh_rigorous_2026}, local controllability \cite{mamon_orbit_2025, francalanci_local_2026} and linear optical invariants \cite{migdal_multiphoton_2014, parellada_lie_2026, draux_invariants_2025, rodari_observation_2025, yang_unveiling_2026} that most photonic states are not in the same orbit (e.g., NOON states for $N>2$ are not in the same orbit as Fock states \cite{parellada_lie_2026}).

In this appendix, we show that the Lie algebraic invariants \cite{parellada_lie_2026} can be used to simplify the exact state preparation problem by optimizing block-diagonal unitaries in $\Um$ instead of arbitrary unitaries. The size of the blocks corresponds to the multiplicity of the invariants. There are two extreme cases. If all the invariants have different values, the blocks are of size 1 and the matrix to optimize is diagonal. If all the invariants have the same value, there is a single block and we gain no improvement. In general, though, it will be something in between. This method speeds up the optimization not only because some matrix computations simplify, but also because the cost function will have fewer local minima, allowing us to find the optimum with fewer trials.

\subsection{Lie algebraic invariants}

In this section, we summarize the Lie algebraic invariants \cite{parellada_lie_2026} that we use to improve the exact state preparation.

Recall from Section \ref{sec: lie groups optics} that the unitary of a passive linear interferometer is mapped to its quantum evolution unitary via a group homomorphism $\varphi$. The differential of this homomorphism, $\dphi$, relates the algebras of both groups (which can be regarded as the effective Hamiltonians of the unitaries \cite{garcia-escartin_multiple_2019}). These relationships can be summarized in a commutative diagram:
\begin{center}
\begin{tikzcd}
    \mathfrak{u}(m) \ni G \quad \arrow{r}{d \varphi}\arrow{d}[swap]{\exp} & \quad  \sum_{ij} {G_{ij}a_i^\dagger a_j} \:\in \mathfrak{u}(D) \arrow{d}{\exp} \\
    \mathrm{U}(m) \ni S=e^G \: \arrow{r}[swap]{ \varphi} &  \: U=e^{d\varphi(G)} \:\in \mathrm{U}(D)
\end{tikzcd}
\end{center}

It is convenient to work with an orthonormal basis of $\mathfrak{u}(m)$:
\begin{align}
    &G^x_{jk}=\frac{i}{\sqrt{2}}(\ket{j}\bra{k}+\ket{k}\bra{j}) &{\rm for}\quad 1 \leq j < k \leq m, \nonumber \\
    &G^y_{jk}=\frac{-1}{\sqrt{2}}(\ket{j}\bra{k}-\ket{k}\bra{j}) &{\rm for}\quad 1 \leq j < k \leq m, \label{eq: basis algebra} \\
    &G^z_j \:=\:i\ket{j}\bra{j}  &{\rm for} \quad j= 1,\ldots , m. \quad \nonumber
\end{align}
Applying $\dphi$ to $\{G_i\}$, we obtain a basis of $\imum$ that we shall call $\{O_i=\dphi(G_i)\}$ (this one is not orthonormal).

Any density matrix $\rho$ corresponding to a mixture of states with $n$ photons in $m$ modes is a Hermitian matrix. Therefore, $i\rho \in \uD$ \footnote{Note that in \cite{parellada_lie_2026} we considered $\uD$ to consist of Hermitian matrices instead of anti-Hermitian ones.} and it can be mapped into the linear optical subalgebra $\imum \subset \uD$ by
\begin{equation}
    \rho_T = \sum_j \tr{i\rho \,O_j} O_j \:.
\end{equation}
This map is equivariant, meaning that, if $\rho$ is evolved into $\sigma=\varphi(S)\rho \varphi(S)^\dagger$, then
\begin{equation}\label{eq: equivariance}
    \sigma_T = (\varphi(S)\rho \varphi(S)^\dagger)_T = \varphi(S)\rho_T \varphi(S)^\dagger \:.
\end{equation}
This equation gives a necessary condition for linear optical evolution. Applying $\dphiinv$, we obtain \cite[Eq. (23)]{parellada_lie_2026}: 
\begin{equation}\label{eq:scattering tangent condition}
    S b_\rho S^\dagger  = b_\sigma \:,
\end{equation}
where
\begin{equation}
    b_\rho = \sum_i \tr{i \rho \, O_j}G_j \:.
\end{equation}
Equation \eqref{eq:scattering tangent condition} gives a simple necessary condition on the scattering matrices that map $\rho$ into $\sigma$, which we can use to simplify the structure of $S$.

\subsection{Simplifying the scattering matrix}

Let $S_0$ be a unitary such that $S_0 b_\rho S_0^\dag = b_\sigma$. Then any $S$ satisfying Eq. \eqref{eq:scattering tangent condition} belongs to the coset of the stabilizer subgroup of $b_\rho$ under the adjoint action of $\Um$ on its algebra, $\um$, \cite{math_exchange}:
\begin{equation} 
\begin{aligned}
    &\left\{S \in \Um  \:\: | \:\: S b_\rho S^\dag = b_\sigma \right\}
    \\
    = &\left\{S_0 \tilde S \:\: | \:\: \tilde S b_\rho \tilde S^\dag = b_\rho \,,\:\:  \tilde S \in \Um  \right\} \:.
\end{aligned}
\end{equation}

We now have two tasks: finding one $S_0$ and finding all the $\tilde S$. To find $S_0$, recall from \cite{parellada_lie_2026} that the eigenvalues of $b_\rho$ are invariant under linear optical evolution, so $b_\rho$ and $b_\sigma$ have the same eigenvalues. Letting $D$ be the diagonal matrix of their eigenvalues in ascending order, we can write $b_\rho = V D V^\dagger$ and $b_\sigma = W D W^\dagger$ for some unitaries $V$ and $W$. We immediately find an $S_0$ such that $b_\sigma = S_0 b_\rho S_0^\dagger$:
\begin{equation}\label{eq:S_0}
    S_0 = W V^\dagger \:.
\end{equation}

All that is left is to find $\tilde S$. From the diagonalization of $b_\rho$ we have $\tilde S V D V^\dagger \tilde S^\dagger = V D V^\dagger$. Therefore, finding all $\tilde S$ is equivalent to finding all unitaries $S' = V^\dagger \tilde S V$ such that
\begin{equation}
    S' D S'^\dagger = D \:.
\end{equation}
Luckily, this simplifies the problem. Since $D$ is diagonal, $S'$ can only be a block-diagonal matrix where the size of each block corresponds to the multiplicity of each eigenvalue. If all eigenvalues are different, $S'$ must be diagonal. If all the eigenvalues are equal, $S'$ can be any unitary in $\Um$. In general, it will be something in between, but it can significantly narrow the search space for an exact state preparation.

\subsection{Simplifying the cost function}

In the previous section, we obtained the decomposition $S = S_0 \tilde S = W S' V^\dag $, where $S'$ is a block-diagonal unitary that we shall optimize. Using the homomorphism property of $\varphi$ \eqref{eq: homomorphism property}, we can decompose $\varphi(S)$:
\begin{equation}
    \varphi(S) = \varphi(W) \varphi(S') \varphi(V^\dagger) \:.
\end{equation}
Note that if $S'$ is block-diagonal, each block acts on a separate subset of modes, so it preserves the number of photons in each subset. Thus, $\varphi(S')$ is also block-diagonal (in a basis where the states sharing the same number of photons in each mode-block are grouped together).

This decomposition applied to the cost function \eqref{eq: cost S} results in a new cost function
\begin{equation}
    \tilde g(S') = \tr{\varphi(S') \tilde B \varphi(S')^\dagger \tilde C} \,,
\end{equation}
where $\tilde B = \varphi(V^\dagger) B  \varphi(V) $ and $\tilde C = \varphi(W)^\dagger C \varphi(W) $. Its gradient with respect to $S'$ is the same as the one we computed in Section \ref{sec: gradient example}, except that it is now a block-diagonal matrix.

\subsection{Examples}

Reducing the optimization from an arbitrary unitary to a block-diagonal unitary provides two main advantages: (1) a higher optimality rate (more random seeds converge to the global optimum) and (2) fewer iterations (and less time) to reach an optimum.

In Table \ref{tab: orbit preparation}, we can see these improvements when searching for a unitary that maps a Fock state into another state in its orbit. For example, the state $\ket{1120}$ has 4 modes, so the basic Riemannian optimization involves finding a $4\times 4$ unitary. In contrast, the Lie algebraic invariants simplify the optimization to a block-diagonal unitary with one block of size 2 (corresponding to $\ket{11}$) and two blocks of size 1 (corresponding to $\ket{2}$ and $\ket{0}$). As we can see, the simplification greatly increases the optimality rate of the optimization and reduces the number of iterations.
\vspace{2mm}

\begin{table}[H]
\centering
\setlength{\tabcolsep}{6pt}
\renewcommand{\arraystretch}{1.2}
\setlength{\aboverulesep}{0.1ex}
\setlength{\belowrulesep}{0.6ex}
\begin{tabular}{cccc}
\toprule
\addlinespace
\textbf{Orbit} & \textbf{Optimizer} & \shortstack{\textbf{Mean}\\\textbf{iters}} & \shortstack{\textbf{Optimality}\\\textbf{rate}}\\
\addlinespace
\midrule
$\ket{1120}$  &  BFGS  & 25  & 22\% \\
$\ket{1120}$  &  BFGS w/ invariants & 4 & 100\% \\
\addlinespace
$\ket{111200}$ &  BFGS  & 37 & 13\%\\
$\ket{111200}$  &  BFGS w/ invariants & 13  & 46\% \\
\addlinespace
\bottomrule
\end{tabular}
\caption{For a given state in the same orbit as the input state, we compare the plain optimization with the optimized version using Lie algebraic invariants. The optimization was considered successful when an infidelity of $10^{-7}$ was reached.}
\label{tab: orbit preparation}
\end{table}

\clearpage

\bibliography{references}

\clearpage

\end{document}